\documentclass[11pt]{article}
\usepackage{graphicx} % Required for inserting images
\usepackage{hyperref}
\usepackage{titletoc}
\usepackage{graphicx}
\usepackage{color} 
\usepackage{amsmath, amsfonts, amssymb, amsxtra, amsthm}
\usepackage{mathrsfs}
\usepackage{dsfont}
\usepackage{verbatim}
\usepackage{comment}
\usepackage{enumitem}
\usepackage{tikz}
\theoremstyle{definition}
\newtheorem{definition}{Definition}
\newtheorem{dn}{Definition}
\newtheorem{theorem}[dn]{Theorem}

\newtheorem{corollary}[dn]{Corollary}
\newtheorem{conjecture}[dn]{Conjecture}

\newtheorem{lemma}[dn]{Lemma}

\usepackage{mathtools}

\newcommand{\dist}{\texttt{dist}}
\newcommand{\eps}{\varepsilon}
\newcommand{\Althofer}{Alth\"{o}fer}
\newcommand{\Erdos}{Erd\"{o}s}
\newcommand{\mst}{\texttt{mst}}

\usepackage{multirow}
\usepackage{makecell}
\usepackage[noadjust]{cite}
\usepackage{mdframed}
\usepackage{placeins}
\usepackage[margin=1in]{geometry}

\usepackage{todonotes}

\title{A Lower Bound for Light Spanners in General Graphs\thanks{This work was supported by NSF:AF 2153680.}}
\author{}
\author{Greg Bodwin and Jeremy Flics\\University of Michigan EECS\\\texttt{\{bodwin,jflics\}@umich.edu}}
\date{}

\begin{document}
\maketitle

\begin{abstract}
A recent upper bound by Le and Solomon [STOC '23] has established that every $n$-node graph has a $(1+\eps)(2k-1)$-spanner with lightness $O(\eps^{-1} n^{1/k})$.
This bound is optimal up to its dependence on $\eps$; the remaining open problem is whether this dependence can be improved or perhaps even removed entirely.

We show that the $\eps$-dependence cannot in fact be completely removed.
For constant $k$ and for $\eps := \Theta(n^{-\frac{1}{2k-1}})$, we show a lower bound on lightness of
$$\Omega\left( \eps^{-1/k} n^{1/k} \right).$$
For example, this implies that there are graphs for which any $3$-spanner has lightness $\Omega(n^{2/3})$, improving on the previous lower bound of $\Omega(n^{1/2})$.

An unusual feature of our lower bound is that it is conditional on the girth conjecture with parameter $k-1$ rather than $k$.
We additionally show that this implies certain technical limitations to improving our lower bound further.
In particular, under the same conditional, generalizing our lower bound to all $\eps$ \emph{or} obtaining an optimal $\eps$-dependence are as hard as proving the girth conjecture for all constant $k$.
\end{abstract}

\thispagestyle{empty}
\setcounter{page}{0}
\pagebreak

\section{Introduction}

We study \emph{spanners}, which are fundamental graph sparsifiers with applications in networking, chip design, flow algorithms, etc.~\cite{ABSHJKS20}

\begin{definition} [Spanners~\cite{PU89jacm, PU89sicomp}]
Given an input graph $G$, a $k$-spanner is a subgraph $H$ satisfying $\dist_H(s, t) \le k \cdot \dist_G(s, t)$ for all nodes $s, t$.
\end{definition}

The goal is generally to construct spanners with a favorable tradeoff between their stretch $k$ and their size.
There are two common ways to quantify the size of a spanner $H$.
The first is by its \emph{sparsity}, that is, the goal is to minimize the number of edges $|E(H)|$.
The general stretch/sparsity tradeoff is well understood, thanks to a classic result of \Althofer{} et al.~\cite{ADDJS93}:
\begin{theorem} [\cite{ADDJS93}] \label{thm:unwtdspan}
For all positive integers $n, k$, every $n$-node graph has a $(2k-1)$-spanner on $O(n^{1+1/k})$ edges.
\end{theorem}

There is also a matching lower bound, conditional on the \emph{girth conjecture} of \Erdos{}.
With respect to a parameter $k$, which may be any positive integer, that conjecture is:
\begin{conjecture} [Girth Conjecture~\cite{girth}]
There exists a family of $n$-node graphs with $\Omega(n^{1+1/k})$ edges and girth $>2k$.
(The \emph{girth} of a graph is the least number of edges in any of its cycles.)
\end{conjecture}
If one removes any edge $(u, v)$ from a girth conjecture graph, there can be no alternate $u \leadsto v$ path of length $\le 2k-1$ (or else the graph has a short cycle).
Thus, any graph $G$ from the girth conjecture has no $(2k-1)$-spanner except for $G$ itself, which means that $G$ gives a matching lower bound to Theorem \ref{thm:unwtdspan}.
The girth conjecture is confirmed for
$$k \in \{1, 2, 3, 5, \Omega(\log n)\},$$
but it is a major open problem to prove or refute it for any other values of $k$ \cite{Wenger91, Tits59}.

The other popular way to measure the size of a spanner $H$ is by its \emph{total edge weight} $w(H)$.
Since the weight of a spanner can be unbounded, one typically normalizes by the weight of a minimum spanning tree (\mst{}) of the input graph.
This measure is called the \emph{lightness} of the spanner.

\begin{definition} [Spanner Lightness]
The \emph{lightness} of a spanner $H$ of a graph $G$ is the quantity
$$\ell(H \mid G) := \frac{w(H)}{w(\mst(G))}.$$
We will also write $\ell(H) := \ell(H \mid H)$ as a shorthand.
\end{definition}

% \noindent In general\footnote{This assumes that $G, H$ are connected; we will frequently assume this in the background.} we have
% $$\ell(H \mid G) \ge \Omega\left( \frac{|E(H)|}{n} \right),$$
% with equality if $H$ is unweighted.
% Thus upper bounds on lightness imply upper bounds on sparsity, and (conversely) lower bounds on sparsity imply lower bounds on lightness.

Lightness has proved much harder to understand than sparsity; after a long line of inquiry (see Table \ref{tbl:priorwork}), the community has only recently obtained a stretch/lightness tradeoff that nearly matches the one known for sparsity, up to an additional $(1+\eps)$ factor in the stretch.
The current upper bound is:
\begin{theorem} [\cite{LS23, Bodwin24}] \label{thm:lightspan}
For all positive integers $n, k$ and all $\eps>0$, every $n$-node graph $G$ has a $(1+\eps) \cdot (2k-1)$-spanner $H$ of lightness
$$\ell(H \mid G) \le O\left(\eps^{-1} \cdot n^{1/k} \right).$$
\end{theorem}

\renewcommand*{\arraystretch}{1.5}

\begin{table}[t]
\begin{center}
\begin{tabular}{llc}
Stretch & Lightness & \textbf{Citation}\\
\hline
$2k-1$ & $O(n/k)$ & \cite{ADDJS93}\\
$(1+\eps) \cdot (2k-1)$ & $O_{\eps}\left(k \cdot n^{1/k}\right)$ & \cite{CDNS92}\\
$(1+\eps) \cdot (2k-1)$ & $O_{\eps}\left(\frac{k}{\log k} \cdot n^{1/k} \right)$ & \cite{ENS15}\\
$(1+\eps) \cdot (2k-1)$ & $O\left( \eps^{-(3+2/k)} n^{1/k} \right)$ & \cite{CW18}\\
$(1+\eps) \cdot (2k-1)$ & $O\left( \eps^{-1} n^{1/k} \right)$ & \cite{LS23} (see also \cite{Bodwin24})\\
\hline
\hline
$2k-1$ & $\Omega\left(n^{1/k}\right)$ & conditional on $\texttt{GC}(k)$ \cite{girth}\\
\makecell[l]{\color{blue}$(1+\eps) \cdot (2k-1)$, for\\\color{blue} fixed $k$ and $\eps = \Theta(n^{-\frac{1}{2k-1}})$} & \color{blue} $\Omega\left( \eps^{-1/k} n^{1/k}\right)$ & \makecell{\color{blue}conditional on $\texttt{GC}(k-1)$\\\color{blue}(this paper)}\\
\end{tabular}
\end{center}
\caption{\label{tbl:priorwork} Summary of prior work on light spanners.  In the lower bounds, $\texttt{GC}(k)$ refers to the girth conjecture with parameter $k$.  See also \cite{FS20} for discussion of existentially optimal spanner algorithms, and \cite{BLW17, BLW19, ADFSW22} and references within for work on light spanners in various special classes of graphs.}
\end{table}

Meanwhile, on the side of lower bounds for lightness, essentially nothing is known.
A girth conjecture graph implies a lower bound of $\Omega(n^{1/k})$ for the lightness of $(2k-1)$-spanners, in exactly the same way as for sparsity.\footnote{By the same argument, an (unweighted) girth conjecture graph $G$ has no nontrivial $(2k-1)$-spanner, and its lightness is $$\ell(G) = \frac{w(G)}{w(\mst(G))} = \frac{|E(G)|}{n-1} = \Omega(n^{1/k}).$$}
But this lower bound does not take advantage of the weighted nature of light spanners, and to date, no lower bounds have been discovered that do leverage edge weights.
So it remains conceivable that \emph{no} $\eps$-dependence is really needed, and that the right bounds for spanner sparsity and lightness are the same, in the sense that girth conjecture graphs are the worst case for each.

The main result of the current paper is that this is not the case, and some $\eps$-dependence in lightness bounds is necessary.
We show:
\begin{theorem} [Main Result] \label{thm:main}
Let $k \ge 2$ be a constant positive integer and assume the girth conjecture with parameter $k-1$.
Then there is a family of $n$-node weighted graphs $G$ for which, letting
$\eps := \Theta(n^{-\frac{1}{2k-1}}),$
any spanner $H$ of stretch $(1+\eps)(2k-1)$ has lightness
$$\ell(H \mid G) \ge \Omega\left(\eps^{-1/k} n^{1/k} \right).$$
\end{theorem}

Thus, some $\eps$-dependence is needed in general.
We can also extend our lower bound to some mildly super-constant $k$, with some additional technical nuances; we discuss these in Section \ref{sec:wrapup}.
To further help interpret our main result, it implies the following corollary:
\begin{corollary} \label{cor:fixedk}
For any positive integer $n$ and constant positive integer $k \ge 2$, assuming the girth conjecture with parameter $k-1$, there exists an $n$-node weighted graph $G$ for which any $(2k-1)$-spanner $H$ has lightness
$$\ell(H \mid G) \ge \Omega\left( n^{\frac{1}{k} + \frac{1}{k(2k-1)}} \right).$$
\end{corollary}

The additional factor of $n^{\frac{1}{k(2k-1)}}$ in this corollary is the part that improves over the previous lower bound from the girth conjecture.
Considering $k=2$ for example, this corollary states that there are graphs on which any $3$-spanner has lightness $\Omega(n^{2/3})$, improving over the previous lower bound from the girth conjecture of $\Omega(n^{1/2})$.

\subsection{Relationship to the Weighted Girth Conjecture}

The technique used by \Althofer{} et al.\ in their proof of Theorem \ref{thm:unwtdspan} is to show that one can always construct a $(2k-1)$-spanner of girth $>2k$, and then to apply the folklore \emph{Moore bound} from extremal graph theory stating that any graph of girth $>2k$ can have only $O(n^{1+1/k})$ edges.
Elkin, Neiman, and Solomon \cite{ENS15} formalized the analogous approach for the light spanner problem by introducing \emph{weighted girth}:

\begin{definition} [Weighted Girth]
For a cycle $C$ in a weighted graph $G = (V, E, w)$, its normalized weight is the quantity
$$w^*(C) := \frac{w(C)}{\max_{e \in C} w(e)}.$$
The weighted girth of $G$ is the minimum value of $w^*(C)$ over all cycles $C$ in $G$.
\end{definition}

They then established the following reduction:
\begin{theorem} [\cite{ENS15}]
For all $t \ge 1$, every graph has a $t$-spanner of weighted girth $>t+1$.
Moreover, any graph of weighted girth $>t+1$ has no nontrivial $t$-spanner.
\end{theorem}

This theorem implies that the question of the lightness needed for a $t$-spanner is completely equivalent to the question of the maximum possible lightness $\ell(H)$ over all graphs of weighted girth $>t+1$.
We will adopt this reframing of the problem in our lower bound construction and analysis.
In particular, to prove Theorem \ref{thm:main}, our goal will be to construct a graph $H$ of weighted girth $>(1+\eps) \cdot 2k$ and lightness
$$\ell(H) \ge \Omega\left( \eps^{-1/k}n^{1/k} \right).$$

Going a bit further, Elkin, Neiman, and Solomon conjectured:
\begin{conjecture} [Weighted Girth Conjecture \cite{ENS15}, c.f.\ Conjecture 1] \label{cjt:wtdgirth}
For all positive integers $n, g$, there is an \textbf{unweighted} $n$-node graph that maximizes lightness over all graphs of weighted girth exactly $g$.
\end{conjecture}

Our new lower bound does \textbf{not} refute this conjecture, and in our view, settling it remains the main open problem in the area of light spanners.
However, we do refute certain stronger versions of the conjecture, which may shed light on how best to attack the weighted girth conjecture going forward.
\begin{corollary} [Refutation of ``Strong'' Weighted Girth Conjectures]
The weighted girth conjecture is false if we instead consider the maximum possible lightness over all graphs of weighted girth strictly larger than $g$, or if we allow $g$ to be non-integral.
\end{corollary}

To further explain our point here, let us focus on the setting $g=4$.
It is known that the densest possible graph of unweighted girth exactly $4$ has $\Theta(n^2)$ edges (the biclique), but that the densest possible graph of unweighted girth $>4$ has $\Theta(n^{3/2})$ edges (the finite projective plane incidence graph \cite{Wenger91}).
Our lower bound implies that there are graphs of weighted girth $>4$ and lightness $\Omega(n^{2/3})$, which beats the lightness bound from the latter construction but not the former.
Thus we refute the ``strong'' weighted girth conjectures quoted above but not the original weighted girth conjecture.
A similar effect holds for any even $g \ge 4$.

\subsection{Limitations to Further Progress and Technical Overview}

The $\eps$-dependence in our lower bound is still a far cry from the dependence in the current upper bounds, and so one might hope to improve it further.
Specifically, one might hope to generalize our lower bound to all $\eps$, or to improve the dependence on $\eps$ to match the dependence in the upper bound.
However, due to the fact that our lower bound relies on the girth conjecture with parameter $k-1$ rather than $k$, it turns out that either of these improvements will be difficult to achieve in the sense that they would imply the girth conjecture itself.

\begin{theorem} [Technical Limitations] \label{thm:techlb}
Suppose that Theorem \ref{thm:main} can be improved in \textbf{either} of the following two ways, while still only conditioning on the girth conjecture with parameter $k-1$:
\begin{itemize}
\item It holds for all $0 < \eps < 1$, rather than just $\eps = \Theta(n^{-\frac{1}{2k-1}})$.
\item It holds with a dependence of $\eps^{-1}$ in the lightness lower bound, rather than $\eps^{-1/k}$.
\end{itemize}
Then the girth conjecture is true for all constant $k$.
\end{theorem}

The girth conjecture has not been settled in any new cases since the '50s \cite{Tits59}, and it is regarded as a major open problem in extremal combinatorics and theoretical computer science.
So this barrier to improving our approach is likely a serious one; it will probably require major new ideas to prove the girth conjecture for all constant $k$ (if this is even possible).

In light of this barrier, we will next do some technical overviewing of why our lower bound conditions on the girth conjecture with parameter $k-1$, and the (informal) conceptual challenge in instead conditioning on parameter $k$.
Perhaps the most natural attempt towards a lower bound for light $(1+\eps)(2k-1)$-spanners would be:
\begin{enumerate}
\item Begin with a girth conjecture graph $G$ with parameter $k$,
\item Increase its edge weights to some value $W$, and
\item Fix an $\mst$ with $\Theta(n)$ nodes and edge weight $1$, and then carefully embed the nodes of $G$ somehow into the nodes of the $\mst$.
\end{enumerate}

The lightness of a graph $G$ obtained from this method would be
$$\ell(G) = \Omega(W n^{1/k}).$$
Thus we would get a meaningful lower bound so long as $W$ is a favorable function of $\eps$, \emph{and} if we can argue that the final graph $G$ does not admit a nontrivial $(1+\eps)(2k-1)$-spanner.
But this seems to be quite difficult.
Even though the \emph{initial} graph $G$ has high (weighted or unweighted) girth and thus has no nontrivial spanner, when we embed it into the $\mst$, we create new paths between its nodes that might ruin its high weighted girth.
So we would need to perform this embedding quite carefully.
This would likely require a nuanced structural understanding of $G$, but since $G$ is conjectured rather than explicit, such a structural understanding seems hard to obtain.
Even in the cases where the girth conjecture is confirmed and we do have explicit constructions of $G$, in particular $k \in \{2, 3, 5\}$, it is still not clear if this embedding can be achieved.

The strategy in this paper is to begin with $G$ as a girth conjecture graph with parameter $k-1$ instead of $k$.
This has advantages and disadvantages relative to the method sketched above.
The main disadvantage is that $G$ will begin with many $2k$-cycles, which we are now responsible for destroying during our construction.
But the main advantage is that $G$ begins with some extra density that we can spend throughout the construction.
We essentially spend this by letting our $\mst$ have $\gg n$ nodes.
This allows us to ensure that no \emph{new} problematic cycles are created in the embedding, roughly by including long \emph{spacers} on the $\mst$ that separate the embedding regions of different nodes of $G$.
It also lets us use a random node-splitting method as we embed to destroy the initial $2k$-cycles in $G$.
The only new structural understanding of girth conjecture graphs required by this method is a bound on the initial number of $2k$-cycles in a parameter $k-1$ girth conjecture graph, which we obtain in Section \ref{sec:gcstructure}.

\section{Lower Bound for Light Spanners}

In the introduction we fixed $k$ to be a constant and $\eps$ to be a particular function of $n$ and $k$.
We will perform most of our construction and analysis with respect to arbitrary $k \le O(\log n)$ and $0 < \eps < \frac{1}{2}$, only discussing further restrictions of these parameters once they arise.
We will assume where convenient that $\eps^{-1}$ is integral, which affects our bounds only by lower-order terms.

\subsection{The Structure of Girth Conjecture Graphs \label{sec:gcstructure}}

Before we begin our main construction and analysis, we will need to establish a few structural facts about girth conjecture graphs.

\begin{lemma} \label{lem:gcstructure}
For any parameter $k$, if the girth conjecture holds, then it can be specifically satisfied by a family of graphs $G$ with all of the following properties:
\begin{enumerate}
\item $G$ is bipartite,
\item $G$ is approximately regular (meaning that all nodes have degree $\Theta(d)$ for some parameter $d$), and
\item For any edge $e \in E(G)$ and any integer $c \ge 1$, the number of $2k+2c$ cycles in $G$ that contain $e$ is at most $O\left(n^{\frac{k+2c-1}{k}}\right)$.
\end{enumerate}
\end{lemma}
% \begin{figure}[t]
% \begin{center}
% \begin{tikzpicture}
%     % Draw the circle    
%     % Define the positions of the nodes
%     \node[circle, fill, inner sep=1.5pt, label=left:$s$] (A) at (160:3cm) {};
%     \node[circle, fill, inner sep=1.5pt] (B) at (180:3cm) {};
%     \node[circle, fill, inner sep=1.5pt, label=left:$t$] (C) at (200:3cm) {};
%     \node[circle, fill, inner sep=1.5pt, label=right:$v$] (V) at (0:3cm) {};

%     \draw [dashed, thick] (A) to node[above, sloped] {at most one $k$-path} (V);
%     \draw [dashed, thick] (C) to node[below, sloped] {at most one $k$-path} (V);

%     % Draw thickened arc between the nodes
%     \draw[ultra thick, blue] (160:3cm) arc (160:200:3cm);
    
%     % Label the arc
%     \node [blue] at (190:3.5cm) {$\pi$};
% \end{tikzpicture}
% \end{center}
% \caption{\label{fig:cyccount} In order to bound the number of $(2k+2c)$-cycles in a girth conjecture graph, we use that there can be at most one $k$-path from $s$ to $v$ and at most one $k$-path from $t$ to $v$.}
% \end{figure}
\begin{proof}
Let $G = (V, E)$ be an arbitrary graph satisfying the girth conjecture.
Enforcing the first two properties (bipartiteness and approximate regularity) is standard, so we recap the proof somewhat briefly here.

First, to enforce bipartiteness, we take a random partition $V = V_1 \cup V_2$ of the vertex set by placing each node in $V_1$ or $V_2$ with equal probability.
Then we delete all edges from $G$ that have both endpoints in $V_1$ or both endpoints in $V_2$.
Each edge survives with probability $1/2$, and so in expectation we remove only half the edges from $E$ in this step, and we have made $G$ bipartite.

Next, to enforce approximate regularity, fix $d = \Theta(n^{1/k})$ as the current average degree of $G$.
While possible, perform either of the following two steps:
\begin{itemize}
\item Find a node of degree $\le d/4$ and delete it from $G$, or
\item Find a node of degree $\ge d$ and split it into two new nodes, with its incident edges equitably partitioned between the new nodes.
\end{itemize}
Since we remove at most $n$ nodes in the first step, we remove at most $n \cdot d/4 \le |E|/2$ edges from $G$ in total, so the total number of edges in $G$ again decreases by at most a constant factor.
When the process is complete the average degree in $G$ is still $\Theta(d)$, with a smaller implicit constant.
Since the number of edges and average degree both change by at most a constant factor, this implies that the number of nodes also changes by only a constant factor, even as nodes are split in the second step.
Thus there are $\Theta(n)$ nodes remaining once the process halts, and so $G$ still satisfies the girth conjecture.

For the last property (counting cycles), let $\pi$ be a path of length $2c$ with $e \in \pi$, and let $v$ be any other node.
The total number of ways to choose $(\pi, v)$ is bounded by
\begin{align*}
&\underbrace{n}_{\text{choose } v} \cdot \underbrace{2c}_{\text{choose position of } e \text{ in } \pi} \cdot \underbrace{O\left(n^{1/k}\right)^{2c-1}}_{\text{extend path by an edge } 2c-1 \text{ times}}\\
\le \ & O\left(cn^{1 + (2c-1)/k}\right)\\
= \ & O\left(cn^{\frac{k+2c-1}{k}}\right).
\end{align*}

We then argue that there can be at most one $(2k+2c)$-cycle that contains $\pi$ as a subpath and $v$ as its antipodal node opposite $\pi$ (with exactly $k$ steps from either endpoint of $\pi$ to $v$).
Let $s, t$ be the endpoint nodes of $\pi$, and notice that any such cycle can be viewed as $\pi$ concatenated with a simple $s \leadsto v$ path of length $k$ and a simple $t \leadsto v$ path of length $k$.
However, there can be at most one simple $s \leadsto v$ (or $t \leadsto v$) path of length $k$, since if there are two such paths then they will imply a cycle of length $\le 2k$, contradicting the girth of $G$.
Thus there can be at most one cycle of length exactly $(2k+2c)$ that contains $\pi$ as a subpath and $v$ as its antipodal node opposite $\pi$, and so our counting of choices for $(\pi, v)$ gives an upper bound for the number of cycles that contain $e$.

Finally, note that any particular $(2k+2c)$-cycle can be witnessed by exactly $2k+2c$ different choices of $(\pi, v)$, since $v$ can be any node on the cycle and $\pi$ is then the path opposite $v$.
So the total number of $(2k+2c)$-cycles is
$$O\left(\frac{c}{c+k} \cdot n^{\frac{k+2c-1}{k}}\right) \le O\left(n^{\frac{k+2c-1}{k}}\right),$$
as claimed.
\end{proof}

\subsection{Lower Bound Construction}

We will next describe our lower bound construction:

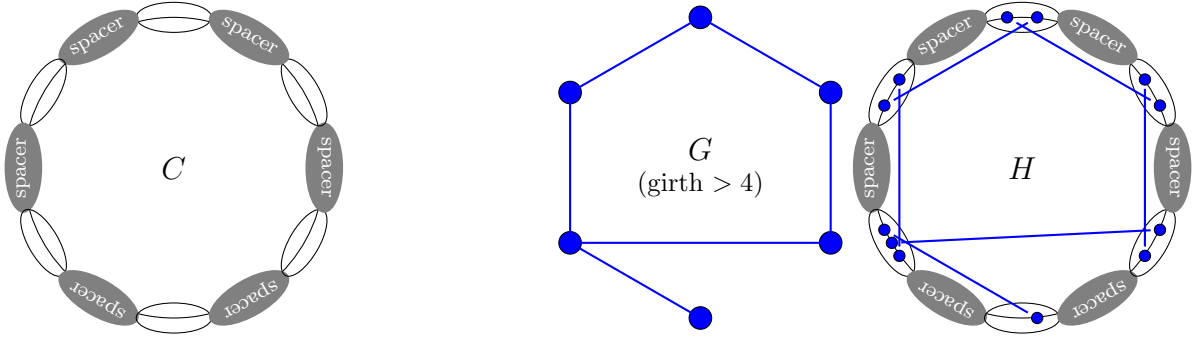
\begin{figure}[h]
\begin{center}
\begin{tikzpicture}
    % Draw the circle
    \draw (0,0) node {\large $C$} circle (2);

    % Draw the non-overlapping ovals enclosing each arc and shade every other oval gray
    \foreach \i [count=\n from 0] in {0,30,...,330} {
        \ifodd\n
            \fill[gray,rotate=\i] (0,2) node [rotate=\i] {\color{white} \footnotesize spacer} ellipse (0.6 and 0.25);
        \else
            \draw[rotate=\i] (0,2) ellipse (0.5 and 0.2);
        \fi
    }

    %draw G
    \begin{scope}[shift={(7, 0)}]
    \node [align=center] at (0, 0) {\large $G$\\ (girth $>4$)};
    \foreach \i [count=\n from 0] in {0,60,...,300} {
        \draw[rotate=\i, fill=blue] (0,2) node (N\n) {} circle [radius=0.15];
    }
    \draw [blue, thick] (N0) -- (N1) -- (N2) -- (N4) -- (N5) -- (N0);
    \draw [blue, thick] (N2) -- (N3);
    \end{scope}
\end{tikzpicture}
\begin{tikzpicture}
    % Draw the circle
    \draw (0,0) node {\large $H$} circle (2);

    % Draw the non-overlapping ovals enclosing each arc and shade every other oval gray
    \foreach \i [count=\n from 0] in {0,30,...,330} {
        \ifodd\n
            \fill[gray,rotate=\i] (0,2) node [rotate=\i] {\color{white} \footnotesize spacer} ellipse (0.6 and 0.25);
        \else
            \draw[rotate=\i] (0,2) ellipse (0.5 and 0.2);
            %draw G nodes
            \draw[rotate=\i, fill=blue] (-0.2,2) node (NL\n) {} circle [radius=0.075];
            \ifthenelse{\n < 5 \OR \n > 6}{
                \draw[rotate=\i, fill=blue] (0.2,2) node (NR\n) {} circle [radius=0.075];
            }{}
            \ifthenelse{\n = 4}{
                \draw[rotate=\i, fill=blue] (0,2) node (NM\n) {} circle [radius=0.075];
            }{}
        \fi
    }

    \draw [blue, thick] (NR0) -- (NL2);
    \draw [blue, thick] (NR2) -- (NL4);
    \draw [blue, thick] (NR4) -- (NL6);
    \draw [blue, thick] (NM4) -- (NL8);
    \draw [blue, thick] (NR8) -- (NL10);
    \draw [blue, thick] (NR10) -- (NL0);
\end{tikzpicture}

\end{center}
\caption{We construct our lower bound graph $H$ by mapping the nodes of a girth conjecture graph (with parameter $k-1$) $G$ into the ``clusters'' of a large cycle $C$, which are separated by ample spacers, and then mapping the edges of $G$ to random edges between clusters.}
\end{figure}

\begin{mdframed} [backgroundcolor=gray!20]
\textbf{Lower Bound Construction.}
\begin{itemize}
\item Let $G$ be an $n$-node graph satisfying the girth conjecture with respect to parameter $k-1$, and with the additional properties listed in Lemma \ref{lem:gcstructure} (bipartite, approximately regular, bounded number of $(2k+2c)$-cycles).

\item Let $C$ be a cycle on $N := 4 k \eps^{-1}n$ nodes in which all edges have weight $1$.
We will refer to $C$ as the \emph{spanning cycle}, or sometimes SC for brevity.

\item Partition $C$ into node-disjoint subpaths that alternately contain $k \eps^{-1}$ nodes (called \emph{clusters}) and $3 k \eps^{-1}$ nodes (called \emph{spacers}).
Note that there are $n$ clusters and $n$ spacers in total.

\item Choose an arbitrary correspondence between the nodes of $G$ and the clusters of $C$; let us write $X_v$ for the cluster assigned to a node $v$.

\item For each edge $(u, v) \in E(G)$, place one edge of weight $\eps^{-1}$ between a uniform-random node pair in $X_u \times X_v$.

\item For each cycle $X$ in $G$ of normalized weight $w^*(X) \le (1+\eps)(2k)$, delete all non-SC edges in $X$ from the graph.
\end{itemize}
We will let $H$ be the final graph.
\end{mdframed}

This construction clearly gives a graph $H$ of weighted girth $> (1+\eps)(2k)$, since in the final step we explicitly remove edge(s) from every cycle of smaller normalized weight (and the spanning cycle itself has normalized weight $N \gg (1+\eps) \cdot 2k$).
The lightness of $H$ is given by:
\begin{lemma} \label{lem:hweight}
Suppose that only a constant fraction of the non-SC edges are removed in the final step of the construction, due to participating in cycles of small normalized weight.
Then the lightness of $H$ is
$$\ell(H) = \Omega\left(\eps^{1/(k-1)} \cdot \frac{N^{1/(k-1)}}{k} \right).$$
\end{lemma}
\begin{proof}
Since we assume that only a constant fraction of the weight is removed in the final step, it suffices to calculate the lightness of $H$ just before this step is performed.
This is:

\begin{align*}
        \ell(H) &:= \frac{w(H)}{w(\mst(H))}\\
                &\ge \Omega\left(\frac{w(H \setminus C)}{w(C)}\right)\\
                &= \Theta\left(\frac{\eps^{-1} |E(G)|}{N}\right)\\
                &= \Theta\left(\frac{\eps^{-1} n^{1+1/(k-1)}}{N}\right)\\
                &= \Theta\left(\frac{\eps^{-1} \left( \frac{\eps N}{k}\right)^{1+1/(k-1)}}{N}\right)\\
                &= \Theta\left(\eps^{1/(k-1)} N^{1/(k-1)} k^{-1}\right). \tag*{\qedhere}
    \end{align*}
\end{proof}

Note that this lightness bound is \emph{increasing} with $\eps$, and so currently we would like to choose $\eps$ as large as possible, since this improves both the weighted girth and the lightness of our construction.
The catch is that we need to satisfy the hypothesis that only a constant fraction of the non-SC edges are removed in the final step, and we will see in the following analysis that this requires us to enforce an upper bound on $\eps$.
Thus, our final setting of $\eps$ will be right at this upper bound.

\subsection{Cycle Analysis \label{sec:cyclea}}

In order to reason about the number of edges that are removed in the final step of our construction of $H$, we need some technical lemmas that inspect the structure of the cycles just before this step.
Let $H'$ be the graph just before the last step of the construction is performed (so $H \subseteq H'$).

\begin{lemma} \label{lem:nwedgecount}
Assuming $\eps < \frac{1}{2}$, every cycle $X$ in $H'$ of normalized weight $w^*(X) \le 2k(1+\eps)$ contains at least $2k$ and at most $2k(1+\eps)$ non-SC edges.
\end{lemma}
\begin{proof}
We show the contrapositive.
If $X$ has $>2k(1+\eps)$ non-SC edges, then its normalized weight is
\begin{align*}
w^*(X) &:= \frac{w(X)}{\max_{e \in X} w(e)}\\
&> \frac{2k(1+\eps) \cdot \eps^{-1}}{\eps^{-1}}\\
&= 2k(1+\eps).
\end{align*}
On the other hand, suppose that $X$ has $<2k$ non-SC edges.
Note that $G$ has girth $\ge 2k$, since (by assumption) it has no cycles of length $\le 2k-2$, and (since it is bipartite) it also has no cycles of length $2k-1$.
Thus the non-SC edges of $X$ do \emph{not} correspond to a cycle in $G$.
That means there are two adjacent non-SC edges $e_1, e_2 \in X$ that do not share a cluster, and so $X$ must include a path of SC edges between them that goes between their respective clusters.
This path includes at least one spacer, and due to this spacer we have
\begin{align*}
w^*(X) &:= \frac{w(X)}{\max_{e \in X} w(e)}\\
&\ge \frac{3k \eps^{-1}}{\eps^{-1}}\\
&> 2k(1+\eps). \tag*{since we assume $\eps < \frac{1}{2}$. \qedhere}
\end{align*}
\end{proof}

In the following, for any cycle $X$ in $G$, we will say that the \emph{corresponding} cycle in $H'$ is the one that contains the image of each edge in $X$, as well as the shortest SC path between any two adjacent edge-images joining their endpoints.

\begin{lemma} \label{lem:cyclesurvival}
Let $c \ge 0$ be an integer and let $X$ be a $(2k+2c)$-cycle in $G$.
Then the probability that the corresponding cycle $X'$ in $H'$ has normalized weight $\le (1+\eps)2k$ is at most
$$\frac{\Theta(\eps)^{2k+2c}}{(2k+2c)!}.$$
\end{lemma}
\begin{proof}
Let $\sigma$ be the total number of SC edges used by $X'$.
The normalized weight of $X'$ is
\begin{align*}
w^*(X') &:= \frac{w(X')}{\max_{e \in X'} w(e)}\\
&= \frac{(2k+2c)\eps^{-1} + \sigma}{\eps^{-1}}\\
&= 2k + 2c + \eps \sigma.
\end{align*}
This is $\le (1+\eps)2k$ iff $\sigma \le 2k - 2c \eps^{-1}$.
So our goal is to bound the probability that the random mapping of the edges of $X'$ into $C$ yields $2k-2c\eps^{-1}$ edges in total, along the SC paths joining their endpoints.

As a first step, let us count the number of ways to partition exactly $2k - 2c \eps^{-1}$ steps among $2k+2c$ clusters, plus a ``remainder'' bucket (to account for the possibility that fewer steps are actually used within clusteres).
This is a standard application of the stars-and-bars counting formula, which gives
$$\binom{2k - 2c \eps^{-1} + 2k + 2c}{2k + 2c} \le \frac{(4k)^{2k+2c}}{(2k+2c)!}$$
ways.
For each such partition, the probability that it is realized by the mapping of $X$ into $H'$ can be bounded by treating the first endpoint mapped to each cluster as fixed, and then acknowledging that there are at most two nodes to which we could map the second endpoint in that cluster that will cause exactly the selected number of in-cluster steps to be used.
So the probability is at most
$$\left(\frac{2}{k \eps^{-1}}\right)^{2k+2c}.$$
Multiplying the previous two terms, we get a bound of
\begin{align*}
\frac{\Theta(\eps)^{2k+2c}}{(2k+2c)!}. \tag*{\qedhere}
\end{align*}
\end{proof}

\begin{lemma} \label{lem:killBigCycle}
For any integer $0 \le c \le k$, if
$$\eps \le O\left( k \cdot n^{-\frac{k+2c}{(k-1)(2k+2c)}} \right)$$
with a small enough implicit constant, then for each non-SC edge $e$, the expected number of cycles $X'$ in $H'$ that have $e \in X'$, exactly $2k+2c$ non-SC edges, and normalized weight $w^*(X') \le (1+\eps)2k$ is at most $\frac{1}{4^{c+1}}$.
\end{lemma}
\begin{proof}
First, let us bound the number of $(2k+2c)$-cycles in $G$.
Recall that $G$ is a girth conjecture graph with parameter $k-1$.
So, instead viewing these as $(2(k-1)+2(c+1))$-cycles and then applying Lemma \ref{lem:gcstructure}, the number of $(2k+2c)$-cycles in $G$ that contain $e$ is bounded by
$$O\left(n^{\frac{k+2c}{k-1}}\right).$$
By Lemma \ref{lem:cyclesurvival}, the probability that each of these cycles in $G$ is mapped to a cycle of small normalized weight in $H'$ is at most
$$\frac{\Theta(\eps)^{2k+2c}}{(2k+2c)!}.$$
We would like to set $\eps$ sufficiently small that the expected number of such cycles is at most $1/4^{c+1}$; that is,
$$O\left(n^{\frac{k+2c}{k-1}}\right) \cdot \frac{\Theta(\eps)^{2k+2c}}{(2k+2c)!} \le \frac{1}{4^{c+1}}.$$
Solving for $\eps$, we get
\begin{align*}
n^{\frac{k+2c}{k-1}} \cdot \Theta(\eps)^{2k+2c} &\le O\left((2k+2c)!\right)\\
\Theta(\eps)^{2k+2c} &\le O\left( (2k+2c)! \cdot n^{-\frac{k+2c}{k-1}} \right)\\
\eps &\le O\left( (2k+2c)!^{\frac{1}{2k+2c}} \cdot n^{-\frac{k+2c}{(k-1)(2k+2c)}} \right)\\
\eps &\le O\left( k \cdot n^{-\frac{k+2c}{(k-1)(2k+2c)}} \right). \tag*{\qedhere}
\end{align*}
\end{proof}

\subsection{Proof Wrapup \label{sec:wrapup}}

The previous lemma dictates our final choice of $\eps$, but note that it expresses its bound as a function of $n$ (the number of nodes in $G$), rather than $N$ (the number of nodes in $H$).
To rearrange the bound as a function of $N$, we calculate:
\begin{align*}
\eps &= \Theta\left( k \cdot n^{-\frac{k+2c}{(k-1)(2k+2c)}} \right)\\
\eps &= \Theta\left( k \cdot \left(\frac{\eps N}{k}\right)^{-\frac{k+2c}{(k-1)(2k+2c)}} \right)\\
\eps^{1 + \frac{k+2c}{(k-1)(2k+2c)}} &= \Theta\left( k \cdot N^{-\frac{k+2c}{(k-1)(2k+2c)}} \cdot k^{\Theta(1/k)}\right)\\
\eps^{\frac{2k^2 + 2kc - k}{(k-1)(2k+2c)}} &= \Theta\left( k \cdot N^{-\frac{k+2c}{(k-1)(2k+2c)}} \right)\\
\eps^{2k^2 + 2kc - k} &= \Theta\left( k^{(k-1)(2k+2c)} \cdot N^{-(k+2c)} \right)\\
\eps &= \Theta\left( k^{\frac{(k-1)(2k+2c)}{{2k^2 + 2kc - k}}} \cdot N^{-\frac{(k+2c)}{2k^2 + 2kc - k}} \right).
\end{align*}

To wrap up the proof, we need to look more carefully at the setting of $c$.
Notice that Lemma \ref{lem:nwedgecount} gives an upper bound on the values of $c$ that we need to consider: that is, so long as $c > k \eps$, there are \emph{no} cycles that use $2k+2c$ non-SC edges and which have weighted girth $\le (1+\eps) \cdot 2k$.
So the appropriate setting of $\eps$ is given by the above formula, with $c := \lfloor k \eps \rfloor$.

Under this setting, for any edge $e$, by Lemma \ref{lem:killBigCycle} the expected number of cycles in $H'$ that contain $e$ and have normalized weight $\le (1+\eps)\cdot 2k$ is bounded by
$$\sum \limits_{i=0}^{\lfloor k \eps \rfloor} \frac{1}{4^{i+1}} \le \frac{1}{2}.$$
By Markov's inequality, this implies that $e$ survives in the final step of our construction with constant probability, and so the conditions of Lemma \ref{lem:hweight} are satisfied and the analysis is complete.

Finally, let us acknowledge that our original main result was phrased more restrictively (constant $k$) but also had a simpler-looking setting of $\eps$.
For fixed $k$, our setting of $\eps$ is $\ll 1/k$, which means that the appropriate setting of $\eps$ is the one that takes $c=0$.
With this restriction in place, we can continue to simplify our setting of $\eps$:
\begin{align*}
\eps &= \Theta\left( k^{\frac{(k-1)(2k)}{{2k^2 - k}}} \cdot N^{-\frac{k}{2k^2 - k}} \right)\\
\eps &= \Theta\left( k \cdot N^{-\frac{1}{2k - 1}} \right),
\end{align*}
which is the setting in our main result discussed in the introduction.
To then calculate our lower bound, by Lemma \ref{lem:hweight} we have
\begin{align*}
\ell(H) &= \Omega_k\left(\eps^{\frac{1}{k-1}} \cdot N^{\frac{1}{k-1}} \right)\\
        &= \Omega_k\left(\left( N^{-\frac{1}{2k-1}}\right)^{\frac{1}{k-1}} \cdot N^{\frac{1}{k-1}} \right)\\
        &= \Omega_k\left(N^{\frac{1}{k}} \cdot N^{\frac{1}{k-1} - \frac{1}{k} - \frac{1}{(2k-1)(k-1)}} \right)\\
        &= \Omega_k\left(N^{\frac{1}{k}} \cdot N^{\frac{1}{k(2k-1)}} \right).
\end{align*}
This bound is the one claimed in Corollary \ref{cor:fixedk}.
Finally, since $\eps = \Theta_k(N^{-\frac{1}{2k-1}})$, we can rewrite this bound as
$$\Omega_k\left(N^{\frac{1}{k}} \cdot \eps^{-\frac{1}{k}} \right),$$
completing the proof.

\section{Proof of Technical Lower Bound}

We now prove Theorem \ref{thm:techlb}, establishing limitations to improving our lower bound further.
For this proof, let $\gamma(n, 2k)$ denote the maximum possible number of edges in an $n$-node graph of girth $>2k$.
So the girth conjecture with parameter $k$ states that

$$\gamma(n, 2k) \ge \Omega(n^{1+1/k}).$$
Prior work on light spanners has established upper bounds that are \emph{independent} of the girth conjecture, in the sense that their lightness is a function of $\gamma$.
Specifically:
\begin{theorem} [\cite{ENS15, LS23}] \label{thm:exoptprior}
For all positive integers $n, k$ and all $\eps > 0$, every $n$-node graph has a $(1+\eps)(2k-1)$ spanner $H$ of lightness
$$\ell(H \mid G) \le O\left( \eps^{-1} \cdot \frac{\gamma(n, 2k)}{n}\right).$$
\end{theorem}

We will use this in our proof of Theorem \ref{thm:techlb}.

\begin{proof} [Proof of Theorem \ref{thm:techlb}]
Our proof of the girth conjecture is by induction on $k$; we may incur constant factors at each level of the induction, since its depth is limited to $k$ and we treat $k$ as a constant.
The base case is $k=2$, where the girth conjecture is already proved \cite{Wenger91}.

In the inductive step, since Theorem \ref{thm:main} is only conditional on the girth conjecture with parameter $k-1$, by the inductive hypothesis we may apply it in our construction to get an (unconditional) lower bound on lightness of 
$$\ell(H) \ge \Omega\left( \eps^{-1/k} n^{1/k} \right).$$
Comparing this lower bound to the upper bound in Theorem \ref{thm:exoptprior}, we must have
$$ \Omega\left( \eps^{-1/k} n^{1/k} \right) \le O\left( \eps^{-1} \cdot \frac{\gamma(n, 2k)}{n}\right).$$
Now, if we assume that Theorem \ref{thm:main} extends to a parameter regime where $\eps$ is a constant (perhaps depending on $k$), then we may set $\eps$ to be a constant and rearrange to get
$$\Omega\left(n^{1+1/k}\right) \le \gamma(n, 2k),$$
implying the girth conjecture for parameter $k$.
Alternately, if we can improve the left-hand side dependence from $\eps^{-1/k}$ to $\eps^{-1}$, then the $\eps$ terms cancel and rearranging in the same way again implies the girth conjecture for parameter $k$, completing the inductive step.
\end{proof}

\bibliographystyle{plain}
\bibliography{refs}
\end{document}